\documentclass[11pt]{article}

\usepackage{style2}

\author[1]{Susanne Albers}
\author[2]{Waldo Gálvez}
\author[3]{Maximilian Janke}

\affil[1]{Department of Computer Science, Technical University of Munich\\
\href{mailto:albers@in.tum.de}{albers@in.tum.de}}

\affil[2]{Department of Computer Science, Technical University of Munich\\
\href{mailto:galvez@in.tum.de}{galvez@in.tum.de}}

\affil[3]{Department of Computer Science, Technical University of Munich\\
\href{mailto:janke@in.tum.de}{janke@in.tum.de}}

\date{}

\title{Machine Covering in the Random-Order Model\footnote{Work supported by the European Research Council, Grant Agreement No. 691672, project APEG.}}

\begin{document}

\maketitle

\begin{abstract}
	
	In the Online Machine Covering problem jobs, defined by their sizes, arrive one by one and have to be assigned to $m$ parallel and identical machines, with the goal of maximizing the load of the least-loaded machine. Unfortunately, the classical model allows only fairly pessimistic performance guarantees: The best possible deterministic ratio of~$m$ is achieved by the Greedy-strategy, and the best known randomized algorithm has competitive ratio $\tilde{O}(\sqrt{m})$ which cannot be improved by more than a logarithmic factor.
	
	Modern results try to mitigate this by studying semi-online models, where additional information about the job sequence is revealed in advance or extra resources are provided to the online algorithm. In this work we study the Machine Covering problem in the recently popular \emph{random-order} model. Here no extra resources are present, but instead the adversary is weakened in that it can only decide upon the input set while jobs are revealed uniformly at random. It is particularly relevant to Machine Covering where lower bounds are usually associated to highly structured input sequences.
	
	We first analyze Graham's Greedy-strategy in this context and establish that its competitive ratio decreases slightly to $\Theta\left(\frac{m}{\log(m)}\right)$ which is asymptotically tight. Then, as our main result,
	we present an improved $\tilde{O}(\sqrt[4]{m})$-competitive algorithm for the problem. This result is achieved by exploiting the extra information coming from the random order of the jobs, using sampling techniques to devise an improved mechanism to distinguish jobs that are relatively large from small ones. We complement this result with a first lower bound showing that no algorithm can have a competitive ratio of $O\left(\frac{\log(m)}{\log\log(m)}\right)$ in the random-order model. This lower bound is achieved by studying a novel variant of the Secretary problem, which could be of independent interest.
\end{abstract}

\section{Introduction}

We study the \emph{Machine Covering} problem, a fundamental load balancing problem where $n$ jobs have to be assigned (or scheduled) onto $m$ identical parallel machines. Each job is characterized by a non-negative size, and 
the goal is to maximize the smallest machine load. This setting is motivated by applications where machines consume resources in order to work, and the goal is to keep the whole system active for as long as possible.  Machine Covering has found additional applications in the sequencing of maintenance actions for aircraft engines~\cite{DF81} and in the design of robust Storage Area Networks~\cite{SSS09}. The offline problem, also known as Santa-Claus or Max-Min Allocation Problem, received quite some research interest, see~\cite{W97,SantaClaus06,MaxMin07} and references therein. In particular, the problem is known to be strongly NP-hard but to allow for a Polynomial-Time Approximation Scheme (PTAS)~\cite{W97}.

This paper focuses on the online version of the problem, where jobs arrive one by one and must be assigned to some machine upon arrival. Lack of knowledge about future jobs can enforce very bad decisions in terms of the quality of the constructed solutions: In a classical lower bound sequence, $m$ jobs of size $1$ arrive first to the system and they must be assigned to $m$ different machines by a competitive deterministic online algorithm. Then subsequent $m-1$ jobs of size $m$ arrive, which make the online algorithm perform poorly, see Figure~\ref{fig:classicalLB}. Indeed, the best possible deterministic algorithm achieves a competitive ratio of $m$~\cite{W97}, and if randomization is allowed, the best known competitive ratio is $\tilde{O}(\sqrt{m})$, which is best possible up to logarithmic factors~\cite{AE98}. The corresponding lower bound also uses that the online algorithm cannot schedule the first $m$ jobs correctly, at least not with probability exceeding $\frac{1}{\sqrt{m}}$.

\begin{figure}[h!]
	\centering
	\resizebox{0.60\textwidth}{!}{\begin{tikzpicture}


\draw[fill=gray!80] (0,0.6) rectangle (9,6.6);

\draw[fill=lightgray!70] (0,0) rectangle (10,0.6);

\draw (0,7.5) -- (0,0) -- (10,0) -- (10,7.5);
\draw (1,0) -- (1,7.5);
\draw (2,0) -- (2,7.5);
\draw (3,0) -- (3,7.5);
\draw (4,0) -- (4,7.5);
\draw (5,0) -- (5,7.5);
\draw (6,0) -- (6,7.5);
\draw (7,0) -- (7,7.5);
\draw (8,0) -- (8,7.5);
\draw (9,0) -- (9,7.5);


\draw[fill=gray!80] (16,0) rectangle (25,6);
\draw[fill=lightgray!70] (25,0) rectangle (26,6);
\draw (25,0.6) -- (26,0.6);
\draw (25,1.2) -- (26,1.2);
\draw (25,1.8) -- (26,1.8);
\draw (25,2.4) -- (26,2.4);
\draw (25,3) -- (26,3);
\draw (25,3.6) -- (26,3.6);
\draw (25,4.2) -- (26,4.2);
\draw (25,4.8) -- (26,4.8);
\draw (25,5.4) -- (26,5.4);

\draw (16,7.5) -- (16,0) -- (26,0) -- (26,7.5);
\draw (17,0) -- (17,7.5);
\draw (18,0) -- (18,7.5);
\draw (19,0) -- (19,7.5);
\draw (20,0) -- (20,7.5);
\draw (21,0) -- (21,7.5);
\draw (22,0) -- (22,7.5);
\draw (23,0) -- (23,7.5);
\draw (24,0) -- (24,7.5);
\draw (25,0) -- (25,7.5);
\end{tikzpicture}}
	\caption{The instance showing that no deterministic algorithm is better than $m$-competitive for Machine Covering. To the left, the best possible solution that an online algorithm can construct achieves minimum load $1$. To the right, the optimal minimum load is $m$.}
	\label{fig:classicalLB}
\end{figure}
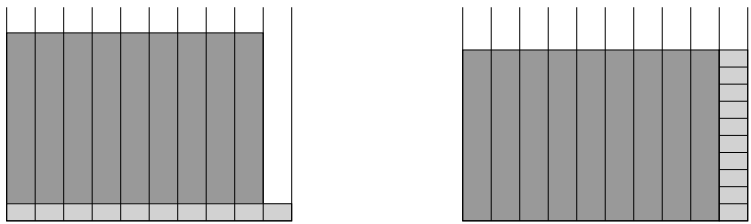

Such restrictive facts have motivated the study of different semi-online models that provide extra information~\cite{AE98,ENSW05,LLMV20,MV20} or extra features~\cite{SSS09, SV16, GSV20, ELS11} to the online algorithm. 

This work studies the Online Machine Covering problem in the increasingly popular \emph{random-order} model. In this model, jobs are still chosen worst possible by the adversary but they are presented to the online algorithm in a uniformly random order. The random-order model derives from the Secretary Problem~\cite{D62, L61} and has been applied to a wide variety of problems such as generalized Secretary problems~\cite{K05, BIKK07, L14, FSZ18, BIKK18}, Scheduling problems~\cite{OT08, M17, AJ20,AJ21}, Packing problems~\cite{K96,KRTV18,AKL19,AKL20}, Facility Location problems~\cite{M01} and Convex Optimization problems~\cite{GMM18} among others. See also~\cite{gupta2020random} for a survey chapter.
It is particularly relevant to Machine Covering, where hard instances force online algorithms to make an irredeemable mistake right on the first $m$ jobs due to some hidden large job class at the end.

Indeed, we show that the competitive ratio of Graham's Greedy-strategy improves from $m$ to $O\left(\frac{m}{\log (m)}\right)$, and that this is asymptotically tight. We also develop an $\tilde{O}(\sqrt[4]{m})$-competitive algorithm, providing evidence that known hardness results rely on ``pathological'' inputs, and complement it by proving that no algorithm can be $O\left(\frac{\log(m)}{\log\log(m)}\right)$-competitive in this model. 

\subsection{Related Results}

The most classical Scheduling problem is \emph{Makespan Minimization} on parallel and identical machines. Here, the goal is dual to Machine Covering; one wants to minimize the maximum load among the machines. This problem is strongly NP-hard and there exists a PTAS~\cite{HS87}. The online setting  received considerable research attention, and already in 1966 Graham showed that his famous Greedy-strategy is $(2-1/m)$-competitive. A long line of research~\cite{GW93, BFKV95, KPT96, A99, FW00} starting in the 1990s lead to the currently best competitive ratio of $1.9201$ due to Fleischer and Wahl~\cite{FW00}. Regarding lower bounds, again after a sequence of results~\cite{FKT89, BKR94, GRTW00} the current best one is $1.88$~\cite{R01}.

The landscape for Online Machine Covering differs considerably from Online Makespan Minimization as discussed before, which has motivated the study of semi-online models to deal with the implied hard restrictions. If the value of the optimal minimum load of the instance is known in advance, Azar~et~al. have shown that a simple greedy algorithm already is $(2-1/m)$-competitive and that no algorithm can attain a competitive ratio better than $7/4$~\cite{AE98}. These bounds were improved by Ebenlendr~et~al. to $11/6$ and $1.791$ respectively~\cite{ENSW05}.
In the \emph{bounded migration} model, whenever a job of size $p$ arrives, older jobs of total size at most $\beta \cdot p$ can be reassigned to different machines. Sanders et~al.~\cite{SSS09} provide a $2$-competitive algorithm for $\beta=1$. Later results~\cite{SV16,GSV20} study the interplay between improved competitive ratios and larger values of $\beta$. 
Another semi-online model provides the online algorithm with a \emph{reordering buffer}, which is used to rearrange the input sequence ``on the fly''. Epstein et~al.~\cite{ELS11} provide a $(H_{m-1}+1)$-competitive algorithm using a buffer of size~$m-1$, and show that this ratio cannot be improved for any sensible buffer size. These and many more semi-online models have also been studied for Makespan Minimization, see the survey in~\cite{E18} and references therein.

The first Scheduling result in the random-order model is due to Osborn and Torng~\cite{OT08}. They establish that the Greedy-strategy for Makespan Minimization does not achieve a competitive ratio better than $2$ for general $m$.  Recently,~\cite{AJ20,AJ21} show that for Makespan Minimization the random-order model allows for better performance guarantees than the classical model. Molinaro~\cite{M17} has studied the Online Load Balancing problem with the objective to minimize general $l_p$-norms of the machine loads, providing an algorithm that returns solutions of expected norm $(1+\varepsilon)\OPT + O\left(\frac{p}{\varepsilon}(m^{1/p}-1)\right)$ in the random-order model, where $\OPT$ denotes the optimal norm. G\"{o}bel et al.~\cite{GKT15} have studied Average Weighted Completion Time Minimization on one machine in the random-order model. Their competitive ratio is logarithmic in the input length $n$ for general job sizes and constant if all jobs have size~$1$. 
To the best of our knowledge, no previous result is known for Online Machine Covering in the random-order model.

\subsection{Our Contribution}
We first establish that the Greedy-strategy is $\Theta\left(\frac{m}{\log(m)}\right)$-competitive in the random-order model. This is only a tiny, albeit significant improvement compared to worst-case orders. Since the bound is tight, more refined strategies that make particular use of the characteristics of the random-order model are required. The analysis also gives first intuitions about what these characteristics are and also about the techniques used to analyze the main algorithm.

The following theorem summarizes the central result of this paper.

\begin{restatable}{theorem}{thmmainapx}\label{thm:main-apx} There exists a $\tilde{O}(\sqrt[4]{m})$-competitive algorithm for the online Machine Covering problem in the random-order model. \end{restatable}

In the classical online Machine Covering problem, difficult instances are usually related to the inability of distinguishing ``small'' and ``large'' jobs induced by a lack of knowledge about large job classes hidden at the end of the sequence. Figure~\ref{fig:classicalLB} depicts the easiest example on which deterministic schedulers cannot perform well as they cannot know that the first $m$ jobs are tiny.
Azar and Epstein~\cite{AE98} ameliorate this by maintaining a randomized threshold, which is used to distinguish small and large job sizes. They have to correctly classify up to $m$ large jobs with constant probability while controlling the total size of incorrectly classified small jobs, which leads to their randomized competitive ratio of $\tilde{O}(\sqrt{m})$. However, their lower bound shows that general randomized algorithms are still unable to schedule the first $m$ jobs correctly with probability exceeding $\frac{1}{\sqrt{m}}$, again due to relevant job classes being hidden at the end of the input. 

Random-order arrival makes such hiding impossible. This already helps the Greedy-strategy as now large jobs in the input are evenly distributed instead of being clustered at the end.
Our $\tilde{O}(\sqrt[4]{m})$-competitive algorithm enhances the path described previously by making explicit use of the no-hiding-feature; it determines those large jobs the adversary would have liked to hide. Information about large job sizes is, as is common in Secretary problems, estimated in a sampling phase, which returns a threshold distinguishing all except for $\sqrt{m}$ of the large jobs. This reduction by a square root carries over to the competitive ratio: We now can allow to misclassify these remaining $\sqrt{m}$ jobs with a higher probability, which in turn leads to a better classification of small jobs and a better competitive ratio of $\tilde{O}(\sqrt[4]{m})$.

We complement the upper bound of $\tilde{O}(\sqrt[4]{m})$ with a lower bound of $\omega\left(\frac{\log(m)}{\log\log(m)}\right)$ for the competitive ratio in the random-order model. Lower bounds in the random-order model are usually considered hard to devise since one cannot hide larger pieces of input. Instead of hiding large job classes, we figuratively make them hard to distinguish by adding noise. To this end, we study a novel variant of the Secretary problem, the \emph{Talent Contest} problem, where the goal is to find a good but not too good candidate (or secretary). 

More in detail, we want to pick the $K$-th best among a randomly permuted input set of candidates. Unlike classical Secretary problems (or the more general Postdoc problem~\cite{V83}), we may pick several candidates as long as they are not better than the $K$-th candidate. Furthermore, we interview candidates $t$ times and make a decision at each arrival. In this setting, information gained by earlier interviews helps the decisions required in later ones. 
It can be proven that the expected number of times the desired candidate can be correctly identified relates to the ability of distinguishing exactly the $m-1$ largest jobs from a Machine Covering instance in the random-order model, and hence bounding the aforementioned expected value allows us to obtain the desired hardness result.

\subsection{Organization of the paper}

In Section~\ref{sec:prelim} we provide the required definitions and tools, and analyze the Greedy-algorithm in the context of random-order arrival. 
In Section~\ref{sec:algorithm} we present our main algorithm and its analysis.  Section~\ref{sec:lowerbound} then introduces the Talent Contest Problem and concludes with a lower bound for the best competitive ratio in the random-order model.

\section{Preliminaries}\label{sec:prelim}

In this section we introduce the main definitions and tools that are used along this work. 


In the Machine Covering problem, we are given $n$ jobs $\CJ=\{J_1,\ldots, J_n\}$, specified by their non-negative sizes $p_i$, which are to be assigned onto $m$ parallel and identical machines. The \emph{load} $l_M$ of a machine $M$ is the sum of the sizes of all the jobs assigned to it. The goal is to maximize the minimum load among the machines, i.e. to maximize $\min_M l_M$.

To an online algorithm, jobs are revealed one-by-one and each has to be assigned permanently and irrevocably before the next one is revealed. Formally, given the symmetric group $S_n$ on $n$ elements, each permutation $\sigma\in S_n$ defines the order in which the elements of $\CJ$ are revealed, namely $\CJ^\sigma=(J_{\sigma(1)},\ldots, J_{\sigma(n)})$. Classically, the performance of an online algorithm $A$ is measured in terms of competitive analysis. That is, if we denote the minimum machine load of $A$\footnote{In case $A$ is randomized, we then refer to the expected minimum load.} on $\CJ^\sigma$ by $A(\CJ^\sigma)$ and the minimum machine load an optimal offline algorithm may achieve by $\OPT(\CJ)$ (which is independent on the order $\sigma$), one is interested in finding a small \emph{(adversarial) competitive ratio} $c=\sup_{\CJ}\sup_{\sigma\in S_n} \frac{\OPT(\CJ)}{A(\CJ^\sigma)}$\footnote{Using the convention that $0/0=1$ and $a/0=\infty$ for $a>0$.}.

In the random-order model, the job order is chosen uniformly at random. We consider the permutation group $S_n$ as a probability space under the uniform distribution. Then, given an input set $\CJ$ of size $n$, we charge $A$ \emph{random-order cost} $A^\rom(\CJ)=\bE_{\sigma\sim S_n}[A(\CJ^\sigma)]=\frac{1}{n!}\sum_{\sigma\in S_n}A(\CJ^\sigma)$. The \emph{competitive ratio in the random-order model} of $A$ is $c=\sup_{\CJ} \frac{\OPT(\CJ)}{A^\rom(\CJ)}$. Throughout this work we will assume that $n$ is known to the algorithm; this assumption is common in the literature and it can be proven that it does not help in the adversarial setting,  see Appendix~\ref{sec.knowingN.p} for details. When clear from the context, we will omit the dependency on $\CJ$.

Given $0\le i \le n$, let $P_i=P_i[\CJ]$ refer to the size of the $i$-th largest job in $\CJ$. Given $i\ge 1$, we define $L_i:=\sum_{j\ge i} P_j$ to be the total size of jobs smaller than $P_i$. Note that the terminology 'smaller' uses an implicit tie breaker since there may be jobs of equal sizes.

\section{Properties and Analysis of the Greedy-strategy}\label{sec:greedy-apx}

We now proceed with some useful properties of Graham's Greedy-strategy. Recall that this algorithm always schedules an incoming job on some least loaded machine breaking ties arbitrarily. The following two lemmas recall useful standard properties of the algorithm, which will help us later to restrict ourselves to simpler special instances. 

\begin{lemma}\label{le.P-mbound}
	The minimum load achieved by the Greedy-strategy is at least $P_m$. 
\end{lemma}

\begin{proof} If the $m$ largest jobs get assigned to different machines, the bound holds directly. On the other hand, if one of the $m$ largest jobs $J_j$ gets assigned to a machine that already had one of the $m$ largest jobs $J_{j'}$, of size $P_{j'}$, then at the arrival of $J_{j}$ the minimum load was at least $P_{j'}\ge P_m$, concluding the claim as the minimum load does not decrease through the iterations. \end{proof}

\begin{lemma}\label{le.Lbound}
	The minimum load achieved by the Greedy-strategy is, for each $i\in \{1,\dots,m\}$, bounded below by $\frac{L_i}{m}-P_{i}$.
\end{lemma}



\begin{proof}
	
	For each machine $M$, let $J_{\textrm{last}}(M)$ be the last job assigned to machine $M$. Let $\CJ_{\textrm{last}}$ be the set of all these last jobs. If $ALG$ denotes the minimum load achieved by the Greedy-strategy, then the load of each machine $M$ in the schedule is at most $ALG+J_{\textrm{last}}(M)$ as jobs are iteratively assigned to a least loaded machine. Now remove the $i-1$ largest jobs in $\CJ_{\textrm{last}}$ from the solution and let $\tilde L$ denote the total size of the remaining jobs in $\CJ_{\textrm{last}}$. Then the total size of all remaining jobs is at most $m\cdot ALG + \tilde{L}$.
	On the other hand, since we only removed $i-1$ jobs, the total size of the remaining jobs is at least $L_i$. Since $\tilde{L} \le (m-i+1)\cdot P_i$, we have that $L_i \le m\cdot ALG + (m-i+1)\cdot P_i$. Hence, the minimum load achieved by the Greedy-strategy is at least $\frac{L_i}{m} - \frac{m-i+1}{m} P_i \ge \frac{L_i}{m} - P_i$. \end{proof}

With these tools we can now prove that the Greedy-strategy has a competitive ratio of at most $O\left(\frac{m}{H_m}\right)$ in the random-order model, where $H_m = \sum_{i=1}^m{\frac{1}{i}}$ denotes the $m$-th harmonic number. We also describe a family of instances showing that this analysis is asymptotically tight.

\begin{theorem}\label{thm:Greedy} The Greedy-strategy is $(2+o_m(1))\frac{m}{H_m}$-competitive in the random-order model, and furthermore this bound cannot be better up to the factor of $2+o_m(1)$. \end{theorem}

\begin{proof} Thanks to Lemma~\ref{le.P-mbound} we can assume that $P_m \le \frac{H_m}{2m} \OPT$. We say that a job is \emph{large} if its size is larger than $\frac{H_m}{2m}\OPT$, otherwise it is small. Let $\tilde{k}<m$ be the number of large jobs in the instance. Note that $L_{\tilde k+1}\ge(m-\tilde{k})\OPT$ since in the optimal solution at most $\tilde k$ machines receive large jobs. Using \Cref{le.Lbound} we may assume $\tilde k\ge m-H_m$, as otherwise we are done.
	
	For a given order $\CJ^\sigma$ and $0\le i \le \tilde{k}$, let $S_i(\CJ^\sigma)$ be the total size of small jobs preceded by precisely $i$ large jobs with respect to $\CJ^\sigma$ (see Figure~\ref{fig:inputsequence}). We will prove that the minimum load achieved by the Greedy-strategy is at least $\min\left\{\sum_{i=0}^{\tilde{k}}{\frac{S_i(\CJ^\sigma)}{m-i}} - \frac{H_m}{2m}\OPT,\frac{H_m}{2m}\right\}\OPT$.
	
	\begin{figure}
		\centering
		\resizebox{\textwidth}{!}{\begin{tikzpicture}


\draw[fill=lightgray!40] (-0.6,0) rectangle (-0.3,0.8);
\draw[fill=lightgray!40] (-0.3,0) rectangle (0.3,0.3);
\draw[fill=lightgray!40] (0,0) rectangle (0.3,0.7);
\draw[fill=lightgray!40] (0.3,0) rectangle (0.6,0.9);
\draw[fill=lightgray!40] (0.6,0) rectangle (0.9,0.4);
\draw[fill=gray!80] (0.9,0) rectangle (1.2,1.5);
\draw[fill=lightgray!40] (1.2,0) rectangle (1.5,0.6);
\draw[fill=lightgray!40] (1.5,0) rectangle (1.8,0.5);
\draw[fill=lightgray!40] (1.8,0) rectangle (2.1,0.8);
\draw[fill=lightgray!40] (2.1,0) rectangle (2.4,0.2);
\draw[fill=gray!80] (2.4,0) rectangle (2.7,2.2);
\draw[fill=lightgray!40] (2.7,0) rectangle (3,0.9);
\draw[fill=lightgray!40] (3,0) rectangle (3.3,0.5);
\draw[fill=gray!80] (3.3,0) rectangle (3.6,1.3);
\draw[fill=lightgray!40] (3.6,0) rectangle (3.9,0.8);
\draw[fill=lightgray!40] (3.9,0) rectangle (4.2,0.7);
\draw[fill=lightgray!40] (4.2,0) rectangle (4.5,0.8);
\draw[fill=lightgray!40] (4.5,0) rectangle (4.8,0.2);
\draw[fill=lightgray!40] (4.8,0) rectangle (5.1,0.4);
\draw[fill=lightgray!40] (5.1,0) rectangle (5.4,0.3);
\draw[fill=gray!80] (5.4,0) rectangle (5.7,1.7);
\draw[fill=lightgray!40] (5.7,0) rectangle (6,0.6);
\draw[fill=lightgray!40] (6,0) rectangle (6.3,0.8);
\draw[fill=lightgray!40] (6.3,0) rectangle (6.6,0.5);
\draw[fill=lightgray!40] (6.6,0) rectangle (6.9,0.4);
\draw[fill=gray!80] (6.9,0) rectangle (7.2,2.4);
\draw[fill=lightgray!40] (7.2,0) rectangle (7.5,0.3);
\draw[fill=lightgray!40] (7.5,0) rectangle (7.8,0.7);
\draw[fill=lightgray!40] (7.8,0) rectangle (8.1,0.6);
\draw[fill=gray!80] (8.1,0) rectangle (8.4,1.1);
\draw[fill=lightgray!40] (8.4,0) rectangle (8.7,0.9);
\draw[fill=lightgray!40] (8.7,0) rectangle (9,0.2);
\draw[fill=lightgray!40] (9,0) rectangle (9.3,0.3);
\draw[fill=lightgray!40] (9.3,0) rectangle (9.6,0.7);
\draw[fill=lightgray!40] (9.6,0) rectangle (9.9,0.5);
\draw[fill=lightgray!40] (9.9,0) rectangle (10.2,0.6);
\draw[fill=gray!80] (10.2,0) rectangle (10.5,2.5);
\draw[fill=lightgray!40] (10.5,0) rectangle (10.8,0.4);
\draw[fill=lightgray!40] (10.8,0) rectangle (11.1,0.8);
\draw[fill=lightgray!40] (11.1,0) rectangle (11.4,0.3);
\draw[fill=lightgray!40] (11.4,0) rectangle (11.7,0.5);
\draw[fill=lightgray!40] (11.7,0) rectangle (12,0.8);


\draw (-0.7,0.5) node[anchor=east] {\footnotesize Order $\CJ^\sigma$:};
\draw [decorate,decoration={brace,amplitude=5pt},xshift=0pt,yshift=0pt]
(0.9,0) -- (-0.6,0) node [anchor=north, midway, yshift=-2.5pt] 
{\scriptsize load $S_0[\CJ^\sigma]$};
\draw [decorate,decoration={brace,amplitude=5pt},xshift=0pt,yshift=0pt]
(2.4,0) -- (1.2,0) node [anchor=north, midway, yshift=-2.5pt] 
{\scriptsize load $S_1[\CJ^\sigma]$};
\draw [decorate,decoration={brace,amplitude=5pt},xshift=0pt,yshift=0pt]
(12,0) -- (10.5,0) node [anchor=north, midway, yshift=-2.5pt] 
{\scriptsize load $S_{\tilde{k}}[\CJ^\sigma]$};

\end{tikzpicture}}
		\caption{A possible order $\CJ^\sigma$ for an instance with $\tilde{k} < m$ large (dark) jobs. They partition the sequence into sets of small (light) jobs, each one having total size $S_i(\CJ^\sigma)$, $i=0,\dots,\tilde{k}$.}
		\label{fig:inputsequence}
	\end{figure}
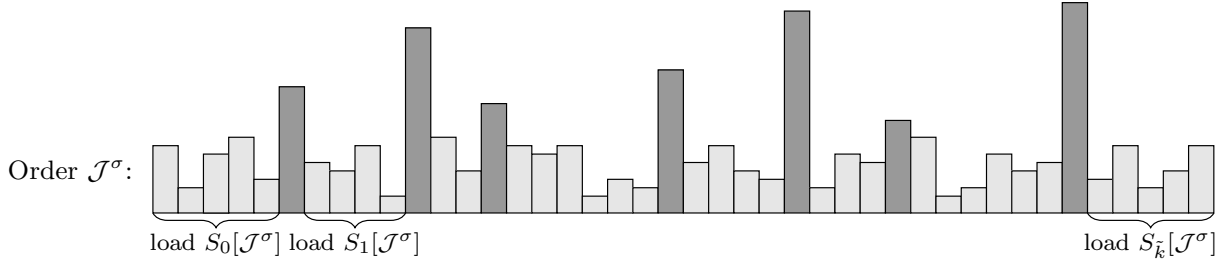
	
	A machine is said to be \emph{full} once it receives a large job, and notice that we can assume that no full machine gets assigned further jobs as otherwise the minimum load would be already at least $\frac{H_m}{2m}\OPT$. Consider the set of small jobs that arrive to the system before the first large job in the sequence. At this point the average load of the machines is exactly $\frac{S_0(\CJ^\sigma)}{m}$ and, since the upcoming large job gets assigned to a least loaded machine, the average load of the remaining machines is still at least $\frac{S_0(\CJ^\sigma)}{m}$. Now the upcoming small jobs that arrive before the second large job in the sequence get assigned only to these non-full machines, whose average load is now at least $\frac{S_0(\CJ^\sigma)}{m} + \frac{S_1(\CJ^\sigma)}{m-1}$. Since the following large job gets assigned to the least loaded of these machines, the average load of the remaining ones is still lower-bounded by this quantity. By iterating this argument, it can be seen that the average load of the machines containing only small jobs is at least $\frac{1}{m-\tilde{k}}\sum_{i=0}^{\tilde{k}}{\frac{S_i(\CJ^\sigma)}{m-i}}$. Since in a Greedy-strategy the load of two machines cannot differ by more than the size of the largest job assigned to them, we conclude that the minimum load achieved by the algorithm is at least $\sum_{i=0}^{\tilde{k}}{\frac{S_i(\CJ^\sigma)}{m-i}} - \frac{H_m}{2m}\OPT$.
	
	Now let us understand the random variable $S(\CJ^\sigma)=\sum_{i=0}^{\tilde{k}}{\frac{S_i(\CJ^\sigma)}{m-i}}$. If we pick $\sigma\sim S_n$ uniformly at random and consider the number $s(J)$ of large jobs that precede any fixed small job $J$, this number is uniformly distributed between $0$ and $\tilde{k}$. Then we can rewrite $S(\CJ^\sigma)=\sum_{i=0}^{\tilde{k}}{\frac{S_i(\CJ^\sigma)}{m-i}}=\sum_{J \text{ small}} \frac{p_J}{m-s(J)}$. Since $\bE[\frac{1}{m-s(J)}]=\sum_{i=0}^{\tilde{k}}{\frac{1}{\tilde k+1}\frac{1}{m-i}}=\frac{H_m - H_{m-\tilde{k}-1}}{\tilde k+1}$ we get by linearity of expectation that \begin{align*} \bE[S(\CJ^\sigma)] &= \sum_{J \text{ small}} p_J\frac{H_m - H_{m-\tilde{k}-1}}{\tilde k+1}\\ &= \frac{H_m - H_{m-\tilde{k}-1}}{\tilde k+1}L_{\tilde k+1}\\ &\ge \frac{ (m-\tilde{k})(H_m - H_{m-\tilde{k}-1})}{\tilde k+1} \frac{L_{\tilde{k}+1}}{m-\tilde{k}}.\end{align*} 
	Notice that $\frac{L_{\tilde{k}+1}}{m-\tilde{k}}$ is a lower bound for $\OPT$, which we will use later. The first factor is decreasing as a function of $\tilde{k}\le m-1$, and consequently the expression is at least $\frac{H_m}{m}$. Thus
	$\bE[S(\CJ^\sigma)]\ge \frac{H_m}{m}\cdot \frac{L_{\tilde{k}+1}}{m-\tilde{k}}.$
	
	We also get $\Var\left[\frac{1}{m-s(J)}\right]\le\bE\left[\left(\frac{1}{m-s(J)}\right)^2\right]=\sum_{i=0}^{\tilde{k}}{\frac{1}{\tilde k+1}\frac{1}{(m-i)^2}}\le \frac{1}{\tilde k+1}\sum_{i} \frac{1}{i^2} \le \frac{1}{\tilde k+1}\frac{\pi^2}{6}$. Using that $\tilde k\ge m - H_m$, we broadly bound $\tilde k$ via $\tilde k+1 \ge \frac{m}{2} \ge \frac{m(m-\tilde k)^2}{2H_m^2}$. Substituting this in the previous bound yields  $\Var\left[\frac{1}{m-s(J)}\right]\le \frac{H_m^2}{m}\frac{1}{(m-\tilde k)^2}\frac{\pi^2}{3}$. Moreover, for two small jobs $J_i$ and $J_j$ we have $\Cov\left[\frac{1}{m-s(J_i)},\frac{1}{m-s(J_j)}\right]\le \left(\Var\left[\frac{1}{m-s(J_i)}\right]\Var\left[\frac{1}{m-s(J_j)}\right]\right)^{1/2} \le \frac{H_m^2}{m}\frac{1}{(m-\tilde k)^2}\frac{\pi^2}{3}$. For $i\neq j$ this bound is pessimistic. The correlation between $s(J_i)$ and $s(J_j)$ is positive but tiny. We use this covariance to bound $\Var[S(\CJ^\sigma)]=\sum_{J \text{ small}} \frac{p_J}{m-s(J)}$.
	
	\begin{eqnarray*}\Var[S(\CJ^\sigma)]&=&\sum_{J_i,J_j \text{ small}} p_i p_j \Cov\left[\frac{1}{m-s(J_i)},\frac{1}{m-s(J_j)}\right] \\
		&\le&\sum_{J_i \text{ small}} p_i \cdot \sum_{J_j \text{ small}} p_j \cdot   \frac{H_m^2}{m}\frac{1}{(m-\tilde k)^2}\frac{\pi^2}{3}\\
		&\le&\frac{\pi^2}{3}\frac{H_m^2}{m}\left(\frac{L_{\tilde k+1}}{m-\tilde k}\right)^2.
	\end{eqnarray*}
	
	The last inequality uses again that $L_{\tilde k+1}\ge (m-\tilde k)\OPT$. Hence, the standard deviation $\text{SD}[S(\CJ^\sigma)]$ is at most $C^{3/2}\frac{L_{\tilde{k}+1}}{m-\tilde{k}}$ for $C=\sqrt[3]{\frac{\pi^2}{3}\frac{H_m^2}{m}}=\Theta\left(\sqrt[3]{\frac{(\log(m))^2}{m}}\right)$. Chebyshev's~inequality, which allows to bound the probability of deviating from the mean in terms of the standard deviation, yields
	\begin{align*}\bP\left[S(\CJ^\sigma) \le \left(\frac{H_m}{m}-C\right)\OPT \right]&\le \bP\left[S(\CJ^\sigma) \le \left(\frac{H_m}{m}-C\right)\frac{L_{\tilde{k}+1}}{m-\tilde{k}} \right]\\
		&\le\bP\left[\left|\bE[S(\CJ^\sigma)]-S(\CJ^\sigma)\right| \ge C^{-1/2}\cdot\sigma[S(\CJ^\sigma)]\right] \\&\le C.\end{align*}
	
	We conclude that with probability $1-C$ the minimum load achieved by the Greedy-strategy exceeds $\min\left\{S(\CJ^\sigma) -\frac{H_m}{2m},\frac{H_m}{2m}\right\}\OPT\ge \frac{H_m}{m}\left(\frac{1}{2}-C\right)\OPT$. Thus its competitive ratio in the random-order model is at most $(1-C)\frac{m}{H_m}\frac{1}{\left(\frac{1}{2}-C\right)}=\left(2+o_m(1)\right)\frac{m}{H_m}$. \end{proof}

\section{An $\boldsymbol{\tilde{O}(\sqrt[4]{m})}$-competitive algorithm}\label{sec:algorithm}

We now describe an improved competitive algorithm for the problem further exploiting the extra features of the random-order model. More in detail, we devise a sampling-based method to classify relatively large jobs (with respect to $\OPT$). After this, we run a slight adaptation of the algorithm \emph{Partition} due to Azar and Epstein~\cite{AE98} in order to distinguish large  jobs that our initial procedure could not classify, leading to strictly better approximation guarantees (see Algorithm~\ref{alg:pseudocode} for a description). We will assume w.l.o.g. that job sizes are rounded down to powers of $2$, 
which induces an extra multiplicative factor of at most $2$ in the competitive ratio.

\subsection{Simple and Proper Inputs}


Given an input sequence $\CJ$, we call any job $J$ \emph{large} if its size is larger than $\frac{\OPT[\CJ]}{100\sqrt[4]{m}}$, and we call it \emph{small} otherwise. Let $k=k[\CJ]$ be the number of large jobs in $\CJ$. Let $L_\msmall=L_{k+1}$ be the total size of small jobs. 
The following definition allows to recognize instances where the Greedy-strategy performs well, see Proposition~\ref{prop:greedyRO}.

\begin{definition}\label{def:simpleinstance}
	We call the input set $\CJ$ \textbf{simple} if either  
	\begin{itemize}
		\item The set $\CJ$ has size $n<m$,
		\item There are at least $m$ large jobs, i.e.\ $k\ge m$, or
		\item There are at most $m-\frac{\sqrt[4]{m^3}}{50}$ large jobs, i.e. $k\le m-\frac{\sqrt[4]{m^3}}{50}$.
	\end{itemize}
\end{definition}

Note that the first condition is mostly included for ease of notation; the third condition implies that sequences with $n< m-\frac{\sqrt[4]{m^3}}{50}$ are simple, which is good enough for our purposes but somewhat clumsy to refer to.

\begin{proposition}\label{prop:greedyRO}
	The Greedy-strategy achieves minimum load at least $\frac{\OPT}{100\sqrt[4]{m}}$ on simple inputs.
\end{proposition}

\begin{proof} 
	
	We consider each case from Definition~\ref{def:simpleinstance} separately:
	\begin{itemize} \item If the instance has less than $m$ jobs, $\OPT=0$ and every algorithm is optimal.
		\item If there are at least $m$ large jobs in the instance, then the minimum load achieved by the Greedy algorithm is at least $P_m \ge \frac{\OPT}{100\sqrt[4]{m}}$ thanks to Lemma~\ref{le.P-mbound}.
		\item If there are at most $m-\frac{\sqrt[4]{m^3}}{50}$ large jobs, then $L_{k+1} \ge \frac{\sqrt[4]{m^3}}{50}\OPT$ as there are at least $\frac{\sqrt[4]{m^3}}{50}$ machines without large jobs in the optimal solution. If we apply Lemma~\ref{le.Lbound} with $i=k+1$, the minimum load achieved by the Greedy algorithm is at least $\frac{\OPT}{100\sqrt[4]{m}}$.\qedhere
\end{itemize} \end{proof}

For an input set $\CJ$ which is not simple, let $d:=\lceil \log_2 (m-k)\rceil$. Note that $0\le d \le \left\lceil\frac{3}{4}\log(m)\right\rceil$. We say that such an instance $\CJ$ is \emph{proper (of degree $d$)}.

Algorithm~\ref{alg:pseudocode} guesses a value $t$ with probability $\Omega(1/\log(m))$ to address in a different way simple instances (case $t=-1$) and proper instances of degree $t$ for each $0\le t\le \left\lceil\frac{3}{4}\log(m)\right\rceil$, at the expense of an extra logarithmic factor in the competitive ratio. By Proposition~\ref{prop:greedyRO}, simple instances are sufficiently handled by the Greedy-strategy.
From now on we will focus on proper instances of degree $d$ and show that the corresponding case when $t=d$ in Algorithm~\ref{alg:pseudocode} returns a $O\left(\sqrt[4]{m}\right)$-approximate solution.

\begin{algorithm}[t]
	\caption{The online Algorithm for Random-Order Machine Covering}\label{alg:pseudocode}
	\hspace*{\algorithmicindent} \textbf{Input:} Job sequence $\CJ^{\sigma}$, $m$ identical parallel machines.
	\begin{algorithmic}[1]
		\State Guess $t \in \{-1,0,1,\dots, \left\lceil\frac{3}{4}\log(m)\right\rceil\}$ uniformly at random.
		\State \textbf{if} $t=-1$, \textbf{then} Run the Greedy-strategy and \textbf{return} the computed solution.
		\State \textbf{else} Partition the machines into $2^t$ \emph{small} machines and $m-2^t$ \emph{large} machines.
		\phase{Sampling.}
		\State Schedule the first $n/8$ jobs iteratively into a least loaded large machine.
		\State Let $P^{\uparrow}$ be the $\left(\frac{m-2^t}{8} - \frac{\sqrt{m}}{2}\right)$-th largest job size among these first $n/8$ jobs.\vspace{3pt}
		\phase{Partition.}
		\State $\tau\gets 0$
		\For{$j=\frac{n}{8}+1, \dots, n$}
		\State \textbf{if} {$p_{\sigma(j)} \ge P^{\uparrow}$}, \textbf{then} Schedule job $J_{\sigma(j)}$ onto a least loaded large machine.
		\State \textbf{if} {$p_{\sigma(j)}> \tau$}, \textbf{then} Update $\tau$ to $p_{\sigma(j)}$ with probability $\frac{1}{9\cdot 2^t\sqrt{m}}$.
		\State \textbf{if} {$p_{\sigma(j)} \le \tau$}, \textbf{then} Schedule job $J_{\sigma(j)}$ onto a least loaded small machine.
		\State \textbf{if} {$p_{\sigma(j)} > \tau$}, \textbf{then} Schedule job $J_{\sigma(j)}$ onto a least loaded large machine.
		\EndFor
		\State \textbf{return} the computed solution.
	\end{algorithmic}
\end{algorithm}

\subsection{Algorithm for Proper Inputs of Degree $\boldsymbol{d}$}

Let $\CJ$ be proper of degree $d$ and assume without loss of generality that its size $n$ is divisible by $8$ (an online algorithm can always simulate up to $7$ extra jobs of size $0$ to reduce to this case). The algorithm, assuming $d$ is guessed correctly, chooses $2^d$ \emph{small} machines $\CM_\msmall$, while the other machines $\CM_\mlarge$ are called \emph{large machines}. The algorithm will always assign the incoming job either to a least loaded small or to a least loaded large machine, the only choice it has to make is to which set of machines the job will go. Theoretically, the goal would be to assign all large jobs to large machines and all small ones to small machines according to our definition. Large and small jobs, unfortunately, cannot be distinguished by an online algorithm with certainty. Instead, we have to use randomization and expect a small error. We aim for a small one-sided error, meaning that we want to avoid misclassifying large jobs at all cost while incorrectly labeling very few small jobs as large.

The algorithm starts with a \emph{sampling phase}: the first $\frac{n}{8}$ jobs will be used for sampling purposes, and since we yet lack good knowledge about what should be considered large, these jobs will all be assigned to large machines. Let $P^\uparrow$ be the element of rank $\frac{m-2^d}{8} - \frac{\sqrt{m}}{2}$ among these elements. The following lemma shows that $P^\uparrow$ can be used as a threshold to distinguish most of the large jobs from small ones.

\begin{lemma}\label{le.parrowbound}
	For proper sequences, it holds that $\bP[P_{k-8\sqrt{m}-2^d} \ge P^\uparrow \ge P_k]\ge \frac{1}{3}.$
\end{lemma}

\begin{proof} Let $n_\mlarge$ be the number of large jobs among the first $\frac{n}{8}$ jobs in the sequence. Notice that $n_\mlarge$ obeys an hypergeometric distribution with parameters $N=n$ (size of the total population), $K=k$ (number of elements with the desired property) and $r=\frac{n}{8}$ (size of the sample). This implies that $\bE_{\sigma\sim S_n}[n_\mlarge] = \frac{k}{8} \ge \frac{m-2^d}{8}$. Moreover \[\Var[n_\mlarge] = \frac{n}{8} \cdot \frac{k}{n} \cdot \frac{n-k}{n} \cdot \frac{7n}{8(n-1)} = \frac{7}{64} \cdot \frac{k(n-k)}{n-1} < \frac{m}{8}.\] If we use Cantelli's inequality, a one-sided version of Chebyshev's inequality, then \[\bP[P^\uparrow < P_k] = \bP\left[n_\mlarge < \frac{m-2^d}{8} - \frac{\sqrt{m}}{2} \right] \le \bP\left[n_\mlarge < \bE[n_\mlarge]-\frac{\sqrt{m}}{2}  \right] \le  \frac{m/8}{\frac{m}{8} + \frac{m}{4} } = \frac{1}{3}.\]
	
	Consider now $n'_\mlarge$ to be the number of the $k-8\sqrt{m}-2^d$ largest jobs among the $\frac{n}{8}$ first jobs in the sequence. Similarly as before, $n'_\mlarge$ obeys an hypergeometric distribution with parameters $N=n$, $K=k-8\sqrt{m}-2^d$ and $r=\frac{n}{8}$, which implies that $\bE_{\sigma\sim S_n}[n'_\mlarge] = \frac{k-2^d-8\sqrt{m}}{8} < \frac{m-2^{d}}{8} -\sqrt{m}$. Here we use that $k\le m$. Furthermore $\Var[n'_\mlarge] = \frac{n}{8} \cdot \frac{k-8\sqrt{m}-2^d}{n} \cdot \frac{n-k+8\sqrt{m}+2^d}{n} \cdot  \frac{7n}{8(n-1)}  \le \frac{m}{8}.$ Then, using again Cantelli's inequality, we obtain that
	\begin{align*} \bP\left[P_{k-8\sqrt{m}-2^d} < P^\uparrow\right] & =  \bP\left[n'_\mlarge \ge \frac{m-2^d}{8} - \frac{\sqrt{m}}{2} \right] \\ 
		&\le \bP\left[n'_\mlarge \ge \bE[n'_\mlarge]+\frac{\sqrt{m}}{2}  \right] \le  \frac{m/8}{\frac{m}{8} + \frac{m}{4} } = \frac{1}{3}. \end{align*} 
	We conclude that $\bP[P_{k-8\sqrt{m}-2^d} \ge P^\uparrow \ge P_{k}]= 1 -   \bP\left[P_{k-8\sqrt{m}-2^d} < P^\uparrow\right]   - \bP[P^\uparrow < P_k] \ge 1/3$. \end{proof}


After the aforementioned sampling phase is finished and $P^\uparrow$ is known, the algorithm will enter a \emph{partition phase}. If a job now has size at least $P^\uparrow$, it is assigned to the large machines as it will be large with high probability. The large jobs below this value are the most difficult to assign. To this end, we define a \emph{threshold value} $\tau$, which we initialize to $0$. If the incoming job has size at most $\tau$ we simply assign it to a least loaded small machine, but whenever we encounter a job $J$ of size $p>\tau$, we set $\tau=p$ with probability $\frac{1}{9\cdot 2^d\sqrt{m}}$. If now $p\le \tau$, which could happen if we just increased $\tau$, we schedule $J$ on the least loaded small machine. Else, $J$ is assigned to the least loaded machine in $\CM_\mlarge$ (see Figure~\ref{fig:algmain} for a depiction of the procedure). As the following lemma shows, this procedure distinguishes all the large jobs with constant probability.

\begin{figure}
	\centering
	\resizebox{\textwidth}{!}{\begin{tikzpicture}[
large/.style={fill=gray},
small/.style={fill=lightgray!50},
misclassified/.style={pattern=north west lines},
sizeline/.style={dashed,line width=2pt}]

%
%

\draw[misclassified] (0.0,0) rectangle (0.15,1.18695036687);
\draw[misclassified] (0.15,0) rectangle (0.3,0.0288375222318);
\draw[large] (0.3,0) rectangle (0.45,2.59211781075);
\draw[misclassified] (0.45,0) rectangle (0.6,0.312689474025);
\draw[misclassified] (0.6,0) rectangle (0.75,0.111917346978);
\draw[misclassified] (0.75,0) rectangle (0.9,0.723006835682);
\draw[large] (0.9,0) rectangle (1.05,2.1717432136);
\draw[misclassified] (1.05,0) rectangle (1.2,1.02844542268);
\draw[misclassified] (1.2,0) rectangle (1.35,0.340804897644);
\draw[misclassified] (1.35,0) rectangle (1.5,0.87744246175);
\draw[misclassified] (1.5,0) rectangle (1.65,1.34337950367);
\draw[misclassified] (1.65,0) rectangle (1.8,0.810801340254);
\draw[misclassified] (1.8,0) rectangle (1.95,1.13611907964);
\draw[small] (1.95,0) rectangle (2.1,0.0475236262469);
\draw[small] (2.1,0) rectangle (2.25,0.0454604966389);
\draw[large] (2.25,0) rectangle (2.4,3.19001015774);
\draw[large] (2.4,0) rectangle (2.55,3.74811993902);
\draw[misclassified] (2.55,0) rectangle (2.7,0.601311928739);
\draw[misclassified] (2.7,0) rectangle (2.85,0.288276476781);
\draw[misclassified] (2.85,0) rectangle (3.0,0.246454421175);
\draw[misclassified] (3.0,0) rectangle (3.15,1.499392398);
\draw[misclassified] (3.15,0) rectangle (3.3,0.324310765756);
\draw[misclassified] (3.3,0) rectangle (3.45,0.548944792726);
\draw[misclassified] (3.45,0) rectangle (3.6,1.45630734735);
\draw[misclassified] (3.6,0) rectangle (3.75,1.03493282153);
\draw[large] (3.75,0) rectangle (3.9,3.59637130239);
\draw[misclassified] (3.9,0) rectangle (4.05,0.323565225133);
\draw[misclassified] (4.05,0) rectangle (4.2,0.173906901944);
\draw[misclassified] (4.2,0) rectangle (4.35,0.746335163586);
\draw[misclassified] (4.35,0) rectangle (4.5,0.0564704544904);
\draw[large] (4.5,0) rectangle (4.65,2.55317912652);
\draw[misclassified] (4.65,0) rectangle (4.8,1.24242558101);
\draw[small] (4.8,0) rectangle (4.95,0.706349894517);
\draw[small] (4.95,0) rectangle (5.1,0.550429766481);
\draw[large] (5.1,0) rectangle (5.25,2.85720708899);
\draw[small] (5.25,0) rectangle (5.4,0.191189629802);
\draw[small] (5.4,0) rectangle (5.55,0.781800829096);
\draw[large] (5.55,0) rectangle (5.7,3.41013411188);
\draw[large] (5.7,0) rectangle (5.85,1.9);
\draw[small] (5.85,0) rectangle (6.0,0.362317004342);
\draw[misclassified] (6.0,0) rectangle (6.15,1.15027206881);
\draw[small] (6.15,0) rectangle (6.3,0.48283021968);
\draw[small] (6.3,0) rectangle (6.45,0.270107331864);
\draw[large] (6.45,0) rectangle (6.6,2);
\draw[large] (6.6,0) rectangle (6.75,3.32900237809);
\draw[small] (6.75,0) rectangle (6.9,0.363822149807);
\draw[misclassified] (6.9,0) rectangle (7.05,1.31318846562);
\draw[misclassified] (7.05,0) rectangle (7.2,0.890296595742);
\draw[misclassified] (7.2,0) rectangle (7.35,0.955838951106);
\draw[misclassified] (7.35,0) rectangle (7.5,1.42100825262);
\draw[small] (7.5,0) rectangle (7.65,0.199761282031);
\draw[small] (7.65,0) rectangle (7.8,0.759952550575);
\draw[small] (7.8,0) rectangle (7.95,1.35508701997);
\draw[small] (7.95,0) rectangle (8.1,0.00280757820988);
\draw[small] (8.1,0) rectangle (8.25,0.59832337764);
\draw[small] (8.25,0) rectangle (8.4,1.29153948058);
\draw[small] (8.4,0) rectangle (8.55,0.588422292582);
\draw[small] (8.55,0) rectangle (8.7,0.710604920984);
\draw[small] (8.7,0) rectangle (8.85,0.324409013877);
\draw[large] (8.85,0) rectangle (9.0,3.53156244973);
\draw[small] (9.0,0) rectangle (9.15,0.910587541377);
\draw[large] (9.15,0) rectangle (9.3,3.81030171333);
\draw[small] (9.3,0) rectangle (9.45,0.154737162782);
\draw[small] (9.45,0) rectangle (9.6,0.00386472262081);
\draw[small] (9.6,0) rectangle (9.75,0.543202502613);
\draw[large] (9.75,0) rectangle (9.9,3.3409824157);
\draw[small] (9.9,0) rectangle (10.05,0.980495127578);
\draw[large] (10.05,0) rectangle (10.2,3.58089184092);
\draw[small] (10.2,0) rectangle (10.35,0.194002143042);
\draw[small] (10.35,0) rectangle (10.5,0.497150424974);
\draw[small] (10.5,0) rectangle (10.65,1.2774660449);
\draw[misclassified] (10.65,0) rectangle (10.8,1.40158275587);
\draw[large] (10.8,0) rectangle (10.95,2.47126357315);
\draw[large] (10.95,0) rectangle (11.1,3.85676413194);
\draw[small] (11.1,0) rectangle (11.25,0.353795394504);
\draw[large] (11.25,0) rectangle (11.4,3.26552776364);
\draw[small] (11.4,0) rectangle (11.55,0.830301757546);
\draw[large] (11.55,0) rectangle (11.7,3.62290929592);
\draw[misclassified] (11.7,0) rectangle (11.85,1.40826993397);
\draw[small] (11.85,0) rectangle (12.0,1.14773350373);
\draw[small] (12.0,0) rectangle (12.15,0.0727466009542);
\draw[small] (12.15,0) rectangle (12.3,1.31455914369);
\draw[small] (12.3,0) rectangle (12.45,1.00319721205);
\draw[small] (12.45,0) rectangle (12.6,0.94164851909);
\draw[small] (12.6,0) rectangle (12.75,1.15245004035);
\draw[small] (12.75,0) rectangle (12.9,0.701531145813);
\draw[small] (12.9,0) rectangle (13.05,0.978063614292);
\draw[small] (13.05,0) rectangle (13.2,0.408897580146);
\draw[small] (13.2,0) rectangle (13.35,0.330031037602);
\draw[misclassified] (13.35,0) rectangle (13.5,1.41338736113);
\draw[small] (13.5,0) rectangle (13.65,0.535340005669);

%
%

\draw[sizeline] (1.8,0.0)--(1.95,0.0);
\draw[sizeline] (1.95,0.0)--(1.95,0.0475236262469);
\draw[sizeline] (1.95,0.0475236262469)--(4.8,0.0475236262469);
\draw[sizeline] (4.8,0.0475236262469)--(4.8,0.706349894517);
\draw[sizeline] (4.8,0.706349894517)--(5.4,0.706349894517);
\draw[sizeline] (5.4,0.706349894517)--(5.4,0.781800829096);
\draw[sizeline] (5.4,0.781800829096)--(7.8,0.781800829096);
\draw[sizeline] (7.8,0.781800829096)--(7.8,1.35508701997);
\draw[sizeline] (7.8,1.35508701997)--(13.8,1.35508701997)node[anchor=west] {$\tau$};
\draw[sizeline] (-0.2,2.1717432136)--(13.8,2.1717432136)node[anchor=west] {$P^\uparrow$};

%
%

\draw[line width=2pt] (1.8,-0.4)--(1.8,4);
\draw [decorate,decoration={brace,amplitude=5pt},xshift=0pt,yshift=0pt]
 (1.8,0) -- (0,0) node [anchor=north, midway, yshift=-2pt]
 {\scriptsize sampling};

\end{tikzpicture} }
	\caption{The classification of large (dark) and small (light and dashed) jobs. During sampling, all small jobs are misclassified (dashed ones). Threshold $P^\uparrow$ classifies large jobs, while threshold $\tau$ classifies small jobs. Jobs in between are conservatively classified as large, since misclassifying large jobs is fatal. Increasing $\tau$ due to a small job is a helpful event as less small jobs will be misclassified. On the other hand, increasing $\tau$ due to a large job below $P^\uparrow$ is a fatal event as large jobs will be misclassified. The choice of $P^\uparrow$ ensures that fatal events are unlikely to happen.
	}\label{fig:algmain}
\end{figure}
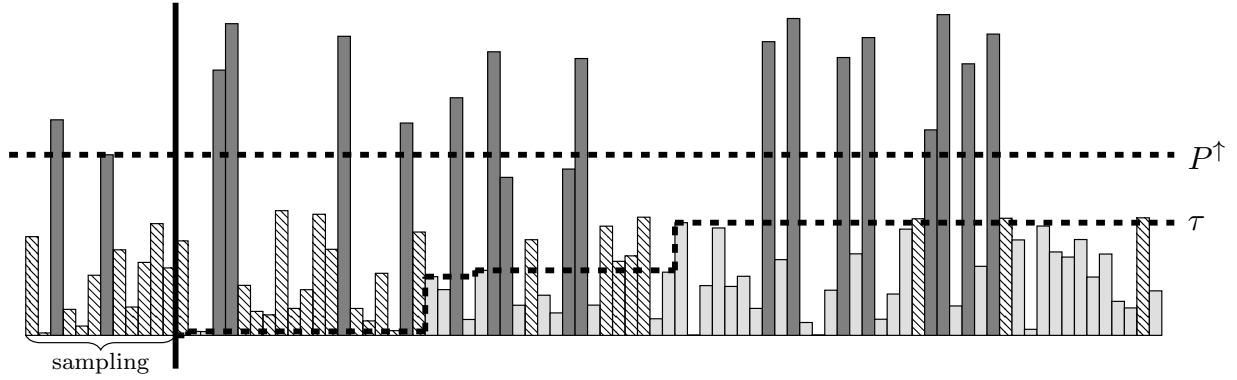

\begin{lemma}\label{le.largejobs}
	All large jobs are scheduled onto large machines with constant probability.
\end{lemma}

\begin{proof}
	Let us assume that $P_{k-8\sqrt{m}-2^d} \ge P^\uparrow \ge P_{k}$. Now, the statement of the lemma can only be wrong if we decided to increase $\tau$ when encountering some large job of size less than~$P^\uparrow$. By assumption there are less than $8\sqrt{m}+2^d$ such jobs. The desired probability is thus at least \[\left(1-\frac{1}{9 \cdot 2^d\sqrt{m}}\right)^{8\sqrt{m}+2^d} \ge 1-\frac{8\sqrt{m}+2^d}{9 \cdot 2^d\sqrt{m}}\ge 1-\frac{8}{9\cdot2}-\frac{1}{9 \sqrt{m}}>\frac{4}{9}.\]
	The first inequality is Bernoulli's inequality, the second one uses that $d\ge 1$. The lemma follows by multiplying with the probability from \Cref{le.parrowbound}.
\end{proof}

We call an input sequence \emph{orderly} if it satisfies the properties of Lemmas~\ref{le.parrowbound} and~\ref{le.largejobs}. The following lemma shows that the total size of misclassified small jobs can be, in expectation, bounded from above. This proof is an adaptation of one of the results from Azar and Epstein~\cite{AE98}, which we present for the sake of completeness.

\begin{lemma}\label{le.smallload}
	The expected total size of the set of small jobs scheduled onto small machines is at least $\frac{31}{800\sqrt[4]{m}}L_\msmall$, even when conditioned on the input sequence being orderly.
\end{lemma}

\begin{proof}
	Let $L'_\msmall$ be the random variable corresponding to the size of small jobs assigned to large machines in the sampling phase. Since jobs appear in random order, we have that \[\bE_{\sigma\sim S_n}[L'_\msmall] = \frac{1}{n!} \displaystyle\sum_{J_i \text{ small}}{ \frac{n}{8}\cdot(n-1)!\cdot p_i} = \frac{1}{8}L_\msmall.\]
	
	Let us now bound the total size of misclassified small jobs in the partition phase. We will define, for a given set of $18\cdot 2^d\sqrt[4]{m}$ small jobs of the same size $p_i$ (recall that jobs are rounded down to powers of $2$), the following event: after $9\cdot 2^d\sqrt[4]{m}$ jobs from the set arrived, $\tau$ is at least $p_i$. If this were not the case, $\tau$ was never updated at any of these  $9\cdot 2^d\sqrt[4]{m}$ first jobs albeit being smaller than $p_i$. The probability for this is at most \[\left(1-\frac{1}{9 \cdot 2^d\sqrt{m}}\right)^{9\cdot 2^d\sqrt[4]{m}} \le e^{-\frac{1}{\sqrt[4]{m}}} \le 1-\frac{1}{2\sqrt[4]{m}},\] where we used the fact that $e^{-x} \le 1-\frac{x}{2}$ if $0\le x \le 1$. This implies that the probability of the previous event occurring is at least $\frac{1}{2\sqrt[4]{m}}$.
	
	Let $\mathcal{S}$ be the set of small jobs remaining after the sampling phase. We will partition $\mathcal{S}$ into batches of $18\sqrt[4]{m}$ jobs of the same size. There will be jobs that are not assigned to any batch because there are not enough jobs of the same size to complete it, but the total size of these jobs is at most \[ 18\cdot 2^d\sqrt[4]{m} \displaystyle\sum_{i\ge 1}{\frac{\OPT}{100\sqrt[4]{m}\cdot 2^i}} \le \frac{18\cdot 2^{d-1} \cdot \OPT}{25} \le \frac{18L_\msmall}{25},\] where the last inequality holds as there are at least $2^{d-1}$ machines in the optimal solution without large jobs. For each of the batches, the probability of assigning at least half of its size to small machines is bounded below by the probability of the previously described event occurring, which is at least $\frac{1}{2\sqrt[4]{m}}$. Hence, the expected total size of small jobs assigned to small machines is at least \[ \frac{1}{2\sqrt[4]{m}}\cdot \frac{1}{2} \displaystyle\sum_{J_i \in \mathcal{S}}{p_i} \ge \frac{1}{4\sqrt[4]{m}}\cdot \left(1-\frac{1}{8}-\frac{18}{25}\right)L_\msmall = \frac{31}{800\sqrt[4]{m}} L_\msmall.\] 
	To observe that we can condition on the sequence being orderly, it suffices to note that the arguments work for every way to fix $P^\uparrow$ and that they do not make any assumptions on $\tau$ being increased at large jobs. \end{proof}

Putting all the previous ingredients together we can conclude the following proposition. 
\begin{proposition}
	The previously described algorithm is $O\left(\sqrt[4]{m}\right)$-competitive in the random-order model for the case of proper inputs of degree $d$.
\end{proposition}

\begin{proof}
	By assumption, all $k\ge m-2^d$ large jobs are scheduled onto large machines. A lower bound of $P_k=\frac{\OPT}{100\sqrt[4]{m}}$ for the minimum load achieved in the large machines follows then from \Cref{le.P-mbound}. By \Cref{le.Lbound}, the minimum load among small machines is at least $\frac{31}{800\sqrt[4]{m}}\frac{L_\msmall}{2^d}-\frac{\OPT}{100\sqrt[4]{m}}$. The proposition follows from observing that $L_\msmall\ge 2^{d-1}\OPT$ since the optimal solution contains at least $m-k\ge 2^{d-1}$ machines with only small jobs.
\end{proof}

\subsection{The Final Algorithm}


As discussed before, our final algorithm first guesses whether the instance is simple or proper of degree $d$. Then we apply the appropriate algorithm, the Greedy-strategy or the previously described algorithm for the right degree. Since there are $O(\log(m))$ many possibilities, this guessing induces an extra logarithmic factor on the final competitive ratio, which concludes the proof of Theorem~\ref{thm:main-apx} restated below.

\thmmainapx*

\section{A Lower Bound for the Random-Order Model}\label{sec:lowerbound}

The main difficulty for Online Machine Covering algorithms, including our main result, is to tell large jobs apart from the largest small jobs. In this section we prove that doing so is, to a certain extent, inherently hard. The main difference to adversarial models is that hardness is not obtained through withholding information but rather through obscuring it. This relates to some studied variants of the classical Secretary Problem such as the Postdoc Problem~\cite{V83} but requires additional features particularly catered to our needs.

\subsection{The Talent Contest Problem}
Consider the following selection problem: To a yearly talent show contest $n$ candidates apply. To appeal to a general audience, we try to exclude the best candidates because an imperfect performance is more entertaining, but we also want to have at least an appropriate candidate who can be presented as the winner. To do so, each candidate will participate in $T$ trials (each trial is considered as an arrival) and we must decide for every arrival if we mark the candidate or not, meaning that we consider her to be the $K$-th best candidate or worse. The global order in which candidates arrive for trials is uniformly distributed, thus at later trials we have much more information to go by. Our final goal is to maximize the number of trials for which we successfully marked the $K$-th best candidate without marking any better candidate.  

Formally, the \emph{Talent Contest} problem is specified by three parameters $K$, $n$ and $T$, where $K\le n$. Candidates have pairwise different non-negative valuations $v_1,v_2,\ldots, v_n$, and each candidate arrives $T$ times; the arrival order is chosen uniformly at random. The valuation of each candidate is revealed when the candidate arrives for the first time. We may decide to mark each arrival or not, though once the next candidate arrives such marking decision is permanent. For each value $1\le h \le T$ we get one point if we marked the $h$-th arrival of the $K$-th best candidate, but not the $h$-th arrival of any better candidate. In particular, we can get up to $T$ points in total, one for each value of $h$.
Let $P(K,T,n)$ be the expected number of points the optimal online algorithm scores given the three values $K,T$ and $n$.  Similar to the classical Secretary problem, we are mostly interested in the limit case $P(K,T)=\lim_{n\rightarrow\infty}P(K,T,n)$.

We require for our desired results one extra technical definition. Given $\lambda \ge 1$, we call the valuations of candidates $\lambda$-steep if all candidate valuations are guaranteed to be at least by a factor $\lambda$ apart, i.e.\ $\min_{v_i > v_j} \frac{v_i}{v_j}\ge \lambda$. It is possible to prove the following bound on the expected value $P(K,T)$, whose proof we defer to \Cref{subsec.proof}.

\begin{lemma}\label{le.FarrierLBSum}
	It holds that $P(K,T)\le \frac{\zeta(T/2)(T+1)^{T/2}}{2\pi \sqrt{K}}$, where $\zeta$ is the Riemann Zeta Function. This bound still holds if we restrict ourselves to $\lambda$-steep valuations for some $\lambda\ge 1$.
\end{lemma}

Roughly speaking, the proof relies on the fact that if an algorithm manages to perform relatively well in the Talent Contest problem, then it could be used to guess the value of a binomially distributed random variable. For the latter problem, difficulty can be directly established. However, the proof involves more technical aspects as some few irregular orders do not allow this reduction and have to be excluded beforehand. The property of $\lambda$-steepness is ensured by choosing $\mu$ sufficiently large and applying the transformation $v\mapsto \mu^v$ to each valuation.


\subsection{Reduction to Machine Covering}

It is possible to show that the Talent Contest problem and the Online Machine Covering problem in the random-order model are related as the following lemma states. 

\begin{lemma}\label{le.transferlemma}
	Given $K$ and $T$, let $m= (K-1) \cdot T +1$. No (possibly randomized) algorithm for Machine Covering in the random-order model on $m$ machines can be better than $\frac{T}{P(K,T)+1}$-competitive.
\end{lemma}

\begin{proof}
	
	Let $\lambda>T$. Consider any instance of the Talent Contest problem with $\lambda$-steep valuations. We will treat the arrival sequence of candidates as a job sequence, where each arrival corresponds to a job of size given by the valuation of the corresponding candidate. We call the $m-1 = T(K-1)$ jobs corresponding to the arrivals of the $K$ largest candidates \emph{large}. The $T$ jobs corresponding to the next candidate are called \emph{medium}. Notice that the size of a medium job is at most $\frac{OPT}{T}$ as evidenced by the schedule that assigns each large job on a separate machine and the $T$ medium jobs onto the single remaining machine. Jobs which are neither large nor medium have total size at most $T \sum_{i=1}^\infty \lambda^{-i}\frac{\OPT}{T} =\frac{\OPT}{\lambda -1}$ and are thus called \emph{small}. They will become negligible for $\lambda\rightarrow\infty$.
	
	Consider an online algorithm $\mathcal{A}_{MC}$ for Machine Covering in the random-order model. We will derive an algorithm $\mathcal{A}_{TC}$ for the Talent Contest problem as follows: $\mathcal{A}_{TC}$ marks each job that gets assigned to a machine that already contains a job of the same size. Let $P$ be the number of points this strategy gets. We will first show that the schedule of $\mathcal{A}_{MC}$ contains a machine which has at most $P+1$ medium jobs and no big job. For this we consider any fixed input order, and if $\mathcal{A}_{MC}$ is randomized, we consider any fixed outcome of its coin tosses.
	
	For $2\le i\le T$, let $w_i$ be an indicator variable that is $1$ if $\mathcal{A}_{TC}$ gains a point for the $i$-th arrival, i.e. if it marks the $i$-th arrival of the $K$-th best candidate but not the $i$-th arrival of a better candidate; $w_i=0$ otherwise. Let also $r_i$ be an indicator variable that is $1$ if $\mathcal{A}_{TC}$ marked the $i$-th arrival of the $K$-th best candidate but still loses due to also marking the $i$-th arrival of a better candidate. Finally, let $\CM_\msmall$ be the machines which do not receive a large job in the schedule of $\mathcal{A}_{MC}$, and let $Z_\mathrm{med}$ be the average number of medium jobs on machines in $\CM_\msmall$. Our intermediate goal is to show that $Z_\mathrm{med}<2 + \sum_{i=2}^T  w_i$. 
	
	Since there are only $m-1$ large jobs, of which at least $\sum_{i=2}^T r_i$ are scheduled on a machine already containing a large job, we have that $|\CM_\msmall|\ge 1+\sum_{i=2}^T r_i$. Let $d\ge 0$ such that $|\CM_\msmall|=1+d+ \sum_{i=2}^T r_i$. 
	Now, observe that $\sum_{i=2}^T  (w_i+r_i)$ counts the number of medium jobs that are placed on a machine already containing a medium job. Thus, the number of medium jobs on machines in $\CM_\msmall$ is at most $|\CM_\msmall|+\sum_{i=2}^T  (w_i+r_i) = 1+d +\sum_{i=2}^T  (w_i+2r_i)$. Now we can bound the average number of medium jobs on $\CM_\msmall$, namely $Z_\mathrm{med}\le \frac{1+d +\sum_{i=2}^T  (w_i+2r_i)}{ 1+d+\sum_{i=2}^T r_i }$. Let us assume that $Z_\mathrm{med}\ge 2$; then the term on the right hand side increases if we set $d$ and all $r_i$ to zero, and we obtain that $Z_\mathrm{med}\le 1 + \sum_{i=2}^T  w_i$ and thus $Z_\mathrm{med}<2 + \sum_{i=2}^T  w_i $. To derive this inequality we assumed that $Z_\mathrm{med}\ge 2$ but it is trivially true if $Z_\mathrm{med}<2$.
	
	Now, let $z_\mathrm{med}$ be the least number of medium jobs on a machine in $\CM_\msmall$. Then $z_\mathrm{med} \le Z_\mathrm{med} < 2 + \sum_{i=2}^T  w_i $. Since $z_\mathrm{med}$ and the right hand side are both integers, it holds that $z_\mathrm{med}\le 1+\sum_{i=2}^T  w_i $. Since $\sum_{i=2}^T  w_i =P$, the number of points obtained by algorithm $\mathcal{A}_{TC}$ for the Talent Contest problem, we have shown that the schedule of $\mathcal{A}_{MC}$ contains a machine $M$ with at most $P+1$ medium jobs and no large one as desired.
	
	As argued before, each medium job has size at most $\frac{\OPT}{T}$ and the small jobs have total size at most $\frac{\OPT}{\lambda -1}$ in total. Thus, machine $M$ has load at most $\frac{P+1}{T}\OPT+\frac{\OPT}{\lambda -1}$. In conclusion, the expected load of the least loaded machine in the schedule of $\mathcal{A}_{MC}$ is at most $\frac{P(K,T)+1}{K}\OPT+\frac{\OPT}{\lambda -1}$, given a worst-case input for the Talent Contest Problem. This concludes the proof by taking $\lambda\rightarrow\infty$.
\end{proof}

By setting $K=(T+1)^T$ and combining this with the lower bound in Lemma~\ref{le.FarrierLBSum}, we obtain the following general lower bound. 

\begin{theorem}\label{le.mainLowerBound}
	The competitive ratio of no online algorithm for Machine Covering in the random-order model, deterministic or randomized, is better than $\frac{\left \lfloor e^{W(\ln(m))}\right\rfloor-1}{1.16+o(1)}$. Here, $W(x)$ is the Lambert W-function, i.e. the inverse to $x \mapsto xe^x$. In particular, no algorithm can be $O\left(\frac{\log(m)}{\log\log(m)}\right)$-competitive for Online Machine Covering in the random-order model.
\end{theorem}

\begin{proof}
	
	Let $K=(T+1)^T$. Then \Cref{le.FarrierLBSum} yields $P(K,T)\le \frac{\zeta(T/2)}{2\pi}\le 1.16 +o(1)$. By \Cref{le.transferlemma} no algorithm can be better than $\left(\frac{T}{1.16+o(1)}\right)$-competitive for $m=(K-1)\cdot T+1 < (T+1)^{T+1}$. We can always add a few jobs of large enough size so that the lower bound extends to larger numbers of machines. The theorem follows since the inverse function of $x\mapsto (T+1)^{T+1}$ is $T\mapsto e^{W(\ln(m))}-1$; the second part uses the identity $W(x)\ge \log(x)-\log\log(x)+\omega(1)$, see \cite{AM08}. \end{proof}

\subsection{Proof of Lemma~\ref{le.FarrierLBSum}}\label{subsec.proof}
We choose the $n$ valuations independently and uniformly at random over the real interval~$[0,n]$. For ease of presentation, we are going to consider the reversed problem, where the best candidate is the one having lowest valuation. We can retrieve the original problem by replacing valuation $v$ with $n-v$ thus reverting their order. Consider any algorithm for this problem. The following result is key to prove our desired claim.

\begin{lemma}\label{le.oldFarrierLB}
	Let $1\le h\le T$, $0<\varepsilon<1/6$  and choose $n$ large enough depending on $\varepsilon$. Let $P_h$ be the probability that algorithm $A$ wins a point for the $h$-th contest, i.e.\ picks the $h$-th arrival of the $K$-th best candidate but not the $h$-th arrival of any better candidate. Then 
	$P_h\le \frac{1}{2\pi \rho^{-T/2}\sqrt{K}}+\varepsilon$ where   $\rho=\frac{T+1-h}{T+1}$. 
\end{lemma}

As argued later, Lemma~\ref{le.FarrierLBSum} follows then by summing up the bound over all possible $h$. 

Consider $h$ fixed. We derive Lemma~\ref{le.oldFarrierLB} by using $A$ to play a certain game, which is easily proven to be difficult. For simplicity we choose $n$ and $\varepsilon$, possibly increasing the former and decreasing the latter, such that $N= (1-\varepsilon)\rho^{-T}n$ is an integer.

We introduce the \emph{binomial guessing game}. We are allowed to pick any point $s$ in the interval $[(1-\varepsilon)K,2K]$ and any value $g$. Afterwards $N= (1-\varepsilon)\rho^{-T}n$ random samples are independently drawn from the interval $[0,n]$. We win if precisely $g$ samples are smaller than~$s$, else we lose.

\begin{lemma}
	For $n$ large enough the best probability with which a player can win the binomial guessing game is bounded above by $\frac{1}{2\pi \rho^{-T/2}\sqrt{K}}+\frac{\varepsilon}{2}$.
\end{lemma}

\begin{proof}
	Let $s$ be the point chosen by the player and let $X$ be the number of the $N$ random samples that are smaller than $s$. Then $X$ is binomially distributed, where the probability of success is precisely $s/n$.
	The optimal strategy of a player for this game is to choose for $g$ the mode, the most likely value, of $X$. This is one of the integers closest to $s/n\cdot N$.  By the Poisson Limit Theorem the probability of $X$ obtaining this (or any!) value is $\frac{g^ge^{-g}}{g!}$ as $n$ and thus $N$ approaches infinity. In this bound we may as well set $g=sN/n$ and ignore the integer rounding, which becomes negligible for $n\rightarrow\infty$. For $n$ large enough, the winning probability is thus at most $\frac{g^ge^{-g}}{g!}+\frac{\varepsilon}{4}$ with $g=sN/n$.
	
	By Stirling's approximation  $\frac{g^ge^{-g}}{g!}+\frac{\varepsilon}{4}\le \frac{1}{2\pi \sqrt{g}}+\frac{\varepsilon}{4}$. Now substitute $g=sN/n\ge (1-\varepsilon)K  \cdot (1-\varepsilon)\rho^{-T} =(1-\varepsilon)^2 \rho^{-T} K$ to obtain that the probability in the statement of the lemma is at most $\frac{1}{(1-\varepsilon) 2\pi \rho^{-T/2}\sqrt{K}}+\frac{\varepsilon}{4}=\frac{1}{2\pi \rho^{-T/2}\sqrt{K}}+\frac{\varepsilon}{(1-\varepsilon) 2\pi \rho^{-T/2}\sqrt{K}}+\frac{\varepsilon}{4} \le \frac{1}{2\pi \rho^{-T/2}\sqrt{K}}+\frac{\varepsilon}{4}+\frac{\varepsilon}{4}$. The last step uses that $K$ and $\rho^{-T/2}$ both exceed $1$ and $(1-\varepsilon) 2\pi>4$.
\end{proof}

Now that we introduced the binomial guessing game, we need to set up the stage for $A$ to play on. Consider all the input permutations in which candidates can arrive. We call such an input permutation \emph{regular} if at least $N+1= (1-\varepsilon)\rho^{-T}n+1$ candidates have not arrived after the $K$-th candidate has arrived $h$ times and additionally the $K$-th best candidate has valuation at least $(1-\varepsilon)K$ and at most $2K$. Otherwise, the input permutation is \emph{irregular}. As the name suggests, we can show that such irregular permutations are very rare.

\begin{lemma}
	The probability that the input permutation is irregular is smaller than $\varepsilon/2$ for~$n$ sufficiently large.
\end{lemma}

\begin{proof}
	
	Let $L$ be the number of candidates that have valuation less than $(1-\varepsilon)K$. A given candidate has valuation below this value with probability $p =(1-\varepsilon)\frac{K}{n}$. Thus, $L$ is binomially distributed with parameters $p$ and $n$. We have $\Var[L]\le\bE[L]=np=(1-\varepsilon)K$. The $K$-th best candidate has valuation at least $(1-\varepsilon)K$ if $L<K$. Note that $K\ge \bE[L]+\varepsilon \Var [L]$,
	and thus Chebyshev's inequality yields $\bP[L\ge K] \le \bP[L\ge \bE[L]+ \varepsilon \Var[L]] \le \varepsilon^2 \le \varepsilon/6$.
	
	Let $U$ be the number of candidates that have valuation less than $2K$. The $K$-th best candidate has valuation at most $2K$ if $U>K$. A similar argument as done for $L$ yields $\bP[U\le K] \le \varepsilon/6$.

	Finally, consider the correct candidate $W$ of rank $K$ and any other candidate~$C$. Candidate~$C$ arrives $T$ times. The probability for each individual arrival to occur after~$W$ arrived $h$ times is precisely $\frac{T+1-h}{T+1}=\rho$. Thus, the probability that candidate $C$ only arrives after $W$ has arrived $h$ times is $\rho^T$.
	Let $Y$ be the total number of candidates arriving after candidate~$W$. Then $Y$ is binomially distributed with $\Var[Y]\le \bE[Y] \le \rho^T(n-1)$. We choose $n$ large enough such that $(1-\varepsilon)\rho^{-T}n+1 \le (1-\varepsilon/2)\rho^{-T}(n-1)$. This implies that, for $Z=\varepsilon/2\cdot \rho^{-T}$, we have $N+1= (1-\varepsilon)\rho^{-T}n+1 \le (1-\varepsilon/2)\rho^{-T}(n-1) \le \bE[Y]-Z\cdot (n-1)\cdot\Var[Y] $.
	Then again, by Chebyshev's inequality, $\bP[Y< N+1]\le \bP[Y\le Z\cdot (n-1)\cdot\Var[Y]] \le \frac{Z}{\sqrt{n-1}}$. This is smaller than $\frac{\varepsilon}{6}$ for $n$ large enough.
	
	The input sequence is irregular only if one of the following holds: $L\ge K$ or $U\le K$ or $Y < N+1$. Thus we have shown that the probability of this event is at most $\varepsilon/6+\varepsilon/6+\varepsilon/6=\varepsilon/2$ for $n$ large enough.
\end{proof}

We use our selection strategy $A$ to play the binomial guessing game. Let $t$ be smallest time at which precisely $N$ candidates are unknown and let us assume that the order of our input sequence is regular. Then the $K$-th candidate has already arrived $h$ times before time~$t$. Consider the smallest candidate picked by $A$ at time $t$, let $s$ be its valuation and $K-g$ its rank among the currently known candidates. Since we are treating a regular order, $A$ loses unless $(1-\varepsilon)K\le s \le 2K$ and precisely $g$ of the $N$ unknown candidates have valuation below $(s,g)$. Since the valuations of these $N$ candidates are chosen according to the uniform distribution on $[0,N]$, algorithm $A$ wins a point for the $h$-th contest if $(s,g)$ wins the binomial guessing game. We can now conclude \Cref{le.oldFarrierLB}.

\begin{proof}[Proof of \Cref{le.oldFarrierLB}]
	Guess $(s,g)$ according to the previously described strategy at time $t$. If $(s,g)$ is not a valid input to the binomial guessing game, pick any value. We argued previously that algorithm $A$ only wins if $(s,g)$ wins the binomial guessing game or if the input order is irregular. The previous two lemmas bound the probability of the former by $\frac{1}{2\pi \rho^{-T/2}\sqrt{K}}+\frac{\varepsilon}{2}$ and the probability of the latter by $\frac{\varepsilon}{2}$.
\end{proof}

So far, we have seen that our selection strategy $A$ will win a point for the $h$-th contest with probability at most $\frac{1}{2\pi \rho^{-T/2}\sqrt{K}}+\varepsilon$. Thus in expectation it wins at most
$ \sum_{h=1}^T \frac{1}{2\pi \rho^{-T/2}\sqrt{K}}+\varepsilon$ points. Since this holds for any selection strategy and since we can choose $\varepsilon$ arbitrarily small we obtain

\begin{align*}
	P(K,T)&\le \sum_{h=1}^T \frac{1}{2\pi \rho^{-T/2}\sqrt{K}}=
	\frac{1}{2\pi \sqrt{K}}\sum_{h=1}^T \frac{(T+1)^{T/2}}{(T+1-h)^{T/2}}\\
	&=\frac{(T+1)^{T/2}}{2\pi \sqrt{K}}\sum_{h=1}^T (T+1-h)^{-T/2}\le \frac{(T+1)^{T/2}}{2\pi \sqrt{K}}\sum_{x=1}^\infty x^{-T/2}=\frac{\zeta(T/2)(T+1)^{T/2}}{2\pi \sqrt{K}}.
\end{align*}

This concludes the proof of first part of Lemma~\ref{le.FarrierLBSum}. To see that we can choose the sequence to be $\lambda$ steep simply replace all the values $v_i$ drawn according to the uniform distribution by $\mu ^{-v_i}$ for  $\mu=\lambda^{n(n-1)/\varepsilon'}$ for some $\varepsilon'$.  Since an online algorithm can also compute this operation, the previous bound still applies to this modified sequence. Moreover, these sequences is almost guaranteed to be $\lambda$-steep.

\begin{lemma}
	The probability that the resulting valuation is not $\lambda$-steep is less that $\varepsilon'$.
\end{lemma}

\begin{proof}
	One can easily verify that the valuation is $\lambda$-steep if (and only if) no two of the $v_i$ picked according to the uniform distribution have distance less that $d=\frac{\varepsilon'}{n(n-1)}$. Consider $1\le i\le n$ and assume that we first pick valuation $v_1\in [0,n]$, then $v_2$ and so on. Then the probability that valuation $v_i$ has distance less than $d$ from any of the previously picked valuations is at most $(i-1)\cdot 2d$. Thus, the probability that there is a pair of valuations at distance less than $d$ is at most $\sum_{i=1}^n (i-1)\cdot 2d = \frac{n(n-1}{2}\cdot 2d = \varepsilon'$.
\end{proof}

Thus we can pick a valuation according to the previous distribution and then replace valuation $v_i$ by $\mu ^{-v_i}$. If this valuation is not $\lambda$-steep we pick any $\lambda$-steep valuation instead. Since the probability of this latter case is at most $\varepsilon'$ the probability of the online algorithm being right increases by at most $\varepsilon'$. The second part of Lemma~\ref{le.FarrierLBSum} follows since we can choose $\varepsilon'$ arbitrarily small, concluding the proof.


\bibliographystyle{plainurl}
\let\oldbibliography\thebibliography
\renewcommand{\thebibliography}[1]{%
  \oldbibliography{#1}%
  \setlength{\itemsep}{0pt plus .3pt}
  \setlength{\parsep}{0pt plus .3pt}
   \setlength{\parskip}{0pt plus .3pt}
}

\bibliography{bibliography}

\appendix

\section{Knowing $n$ does not help in the adversarial model.}\label{sec.knowingN.p}
As common in the area of online algorithms, we ignore computability issues. Formally, a \emph{deterministic online algorithm that does not know $n$} is a function, which maps any input sequence $J_1,\ldots,J_{n'}$ to a machine $M$ onto which the job $J_{n'}$ is scheduled. A \emph{deterministic online algorithm that knows $n$} maps tuples  $n, J_1,\ldots,J_{n'}$ to the machine $M$ onto which job $J_{n'}$ is scheduled. Here $n\ge n'$ is the total length of the input.

Randomized algorithms are simply probability distributions over deterministic algorithms.

\begin{proposition}\label{pro.5.adversial}
	Let $A$ be any (possibly randomized) online algorithm that is $c$-competitive in the adversarial model knowing~$n$ beforehand. Then there is a strategy $B$ that is $c$-competitive not knowing $n$.
	If $A$ is deterministic, the strategy $B$ can be chosen to be deterministic, too.
\end{proposition}
\begin{proof}
	Let us denote the algorithm $A$ that knows beforehand that the input will have length $n$ by $A[n]$. This algorithm outputs schedules for any sequence of jobs having length at most $n$. Given some number $n'\in \BN$, let $\CS[n']$ be the set of all probability distributions over possible schedules of $n'$ jobs. Schedules of $n'$ jobs are nothing but assignments $[n']=\{1,\ldots n'\}\rightarrow \{1,\ldots m\}=[m]$. Let $F$ be the set of all such maps. We can then identify $\CS[n']$ with a compact subspace of $[0,1]^{F}\subset \BR^F$. In fact it is the standard simplex spanned by the set $\CE[n']\subset \CS[n']$ consisting of probability distributions that choose one schedule with probability~$1$, or, geometrically speaking, the set of standard base vectors in $\BR^{M}$. For $n'>0$, forgetting how we scheduled the last job yields a canonical continuous maps $r_{n'}\colon \CS[n']\rightarrow \CS[n'-1]$.
	
	Given jobs $J_1,\ldots J_{n'}$ and $n\ge n'$, let $s_{n'}^n=s_{n'}^n(J_1,\ldots J_{n'})\in\CS[n']$ be the probability distribution of schedules obtained by algorithm $A[n]$ on the job sequence $J_1,\ldots, J_{n'}$. Consider the set of points $s_{n'}^n$ with $n\ge n'$. We pick inductively a limit point $s_{n'}$ of this point set such that $r_{n'}(s_{n'})=s_{n'-1}$. This is possible because $ r_{n'}$ is continuous. The choice of $s_{n'}$ only depends on the jobs $ J_1,\ldots J_{n'}$, hence can be picked 'online'. Now all our algorithm $B$ has to do is to maintain its schedule at time $n'$ according to the probability distribution $s_{n'}$. The condition $r_{n'}(s_{n'})=s_{n'-1}$ ensures that this is possible.
	
	This algorithm is $c$-competitive. Pick $\varepsilon>0$ and any input sequence $J_1,\ldots J_{n'}$. As $s_{n'}$ is a limit point of the $s_{n'}^n$ there exists $n\ge n'$ such that $(1-\varepsilon)s_{n'}^n\le s_{n'}$. In other words, if algorithm $A[n]$ obtains schedule $F$ after treating $J_1,\ldots J_{n'}$ with probability $p$, then $B$ obtains the same schedule with probability at least $(1-\varepsilon)p$. In particular, the expected minimum load of $B$ on $J_1,\ldots J_{n'}$ is at least $(1-\varepsilon)$ times the minimum load of $A[n]$ on that sequence. The latter minimum load is at least $ c\cdot\OPT(J_1,\ldots  J_{n'})$ since $A[n]$ is $c$-competitive. We have shown that the expected minimum load of $B$ is at least $(1-\varepsilon)c\cdot\OPT$. The proposition follows by taking the limit $\varepsilon\rightarrow 0$.
	
	If $A$ is deterministic, all the points $s_{n'}^n$ lie in the closed, in fact discrete, subset $\CE[n']\subset \CS[n']$ of probability distributions simply picking one schedule with probability $1$. All limit points, in particular the points $s_{n'}$, lie in $\CE[n']$ too. Thus, $B$ will also be deterministic.
\end{proof}

\end{document}